\UseRawInputEncoding
\documentclass[journal]{IEEEtran}

\ifCLASSINFOpdf
\else
   \usepackage[dvips]{graphicx}
\fi
\usepackage{url}

\usepackage{graphicx}
\usepackage{cite}
\usepackage{amsmath}
\usepackage{amsfonts}
\usepackage{bm}
\usepackage{caption}
\usepackage{subcaption}
\usepackage{verbatim}
\usepackage{romannum}
\usepackage{xcolor}

\usepackage[]{algorithm2e}
\usepackage{amsthm}
\newtheorem{theorem}{Theorem}

\begin{document}

\title{Estimation of Source and Receiver Positions, Room Geometry and Reflection Coefficients From a Single Room Impulse Response}

\author{Wangyang~Yu, W. Bastiaan~Kleijn,~\IEEEmembership{Fellow,~IEEE}
\thanks{Wangyang~Yu  is with the Dept. of Microelectronics, Delft University of Technology, The Netherlands. Email: W.Yu-1@tudelft.nl.}
\thanks{W. Bastiaan~Kleijn is with the Dept. of Microelectronics, Delft University of Technology, The Netherlands, and with the Faculty of Engineering, Victoria University of Wellington, New Zealand. Email: w.b.kleijn@tudelft.nl.}}

\maketitle

\begin{abstract}
We propose an algorithm to estimate source and receiver positions, room geometry and reflection coefficients from a single room impulse response simultaneously. It is based on a symmetry analysis of the room impulse response. The proposed method utilizes the times of arrivals of the direct path, first order reflections and second order reflections. The proposed method is robust to erroneous pulses and non-specular reflections. It can be applied to any room with parallel walls as long as the required arrival times of reflections are available. In contrast to the state-of-art method, we do not restrict the location of source and receiver.
\end{abstract}


\IEEEpeerreviewmaketitle

\section{Introduction}

\IEEEPARstart{A}{ccurate} acoustic environment modelling forms an important aspect of room acoustics. It has a variety of applications such as speech enhancement \cite{397090, loizou2013speech, HU2007588}, speech recognition \cite{4518676, LIU201565, 6334330}, and sound rendering \cite{4042282, merimaa2005spatial, 7005804}. The room impulse response (RIR), the transfer function between the source and receiver, characterises the room acoustic environment. It is affected by a set of room acoustic attributes including room geometry, the positions of source and receiver, and the reflection coefficients. In this paper, we aim to analyse RIRs and derive the room geometry, the positions of source and receiver, and the reflection coefficients given a single RIR.

The image source method \cite{doi:10.1121/1.382599, kuttruff2014room,begault20003, MCGOVERN2009182, habets2014} is an efficient and widely used RIR modelling method, which was first proposed in \cite{doi:10.1121/1.382599}. This method can be applied in empty rectangular rooms. It is assumed that the sound propagates along straight lines. We use $\mathbf{p, m}$ to label each reflection where each element of $\mathbf{p} = (q, j, k)$ can take a value of $0$ or $1$, indicating the direction of the reflection, and each element of $\mathbf{m}= (m_x, m_y, m_z)$ can take an integer value, associated with the reflection order $O_{\mathbf{p,m}}$ as
\begin{equation}
O_{\mathbf{p,m}} = |2m_x-q| + |2m_y-j| + |2m_z-k|.
\end{equation}
Let $\beta_{x_1},\beta_{x_2},\beta_{y_1},\beta_{y_2},\beta_{z_1},\beta_{z_2}$ denote the six reflection coefficients of the walls, $\tau_{\mathbf{p, m}}$ denote the time of arrival (TOA) of each reflection, and $d_{\mathbf{p, m}}$ denote corresponding path length. Then with the image source method, assuming that the reflection coefficients are finite and constant over each wall, the RIR can be written as 
\begin{equation}\label{RIRIS}\scriptsize
h(t) = \sum_\mathbf{p, m} \beta_{x_1}^{|m_x-q|} \beta_{x_2}^{|m_x|} \beta_{y_1}^{|m_y-j|} \beta_{y_2}^{|m_y|} \beta_{z_1}^{|m_z-k|} \beta_{z_2}^{|m_z|}\frac{\delta(t-\tau_{\mathbf{p, m}})}{4\pi d_{\mathbf{p, m}}.}
\end{equation}

A number of algorithms to estimate the room geometry from RIRs exist \cite{Dokmani12186,7471691, 6811439, 8214218}. \cite{Dokmani12186} estimates the room geometry by exploiting the properties of Euclidean distance matrices. It requires a single source and four receivers whose pairwise distances are known.  The room geometry is estimated with one source and five receivers in \cite{7471691} and the method can achieve approximately $1$ cm accuracy. \cite{6811439} adopts a source-centered coordinate system and uses a single RIR to estimate the room geometry. However, it requires the TOAs to be labelled to the corresponding image source. \cite{8214218} infers room geometry with sets of TOAs from RIRs between one source and an array of receivers with knowledge of the receiver array geometry. The sets of TOAs are detected and labelled to estimate the positions of source and image sources, which can be used to estimate the room geometry afterward. The RIR can also be used for indoor localisation \cite{6853834, 7394316, 5427443}. \cite{6853834} locates the source with a single RIR with Euclidean distance matrices by assuming that both the room geometry and the receiver position are known. \cite{7394316} proposes a source localisation method using the parameters that are extracted from RIRs with an ad-hoc microphone array. The RIRs can also be used for localisation based on the fingerprint \cite{5427443}. \cite{5946405} estimates the room geometry and the source position using the first and second order reflections from one RIR. In contrast to the present work, it assumes that the source and receiver are co-located. Our previous paper estimates room geometry and reflection coefficients from a single RIR using deep neural networks \cite{9286412}.

Our contribution is the introduction of an analytical method to estimate the room geometry, the positions of source and receivers, and reflection coefficients simultaneously using a single RIR without any prior information. The method applies to rooms with parallel wall pairs. The paper is organised as follows. The algorithm to estimate room acoustical parameters is described in Section \Romannum{2}. In Section \Romannum{3}, we analyse a RIR focusing on degeneracy. We discuss our experiment in Section \Romannum{4} and conclude the paper in Section \Romannum{5}.

\section{Room acoustical parameters estimation}
\label{sec:paraestimate}
In this section, we propose a method to determine a particular valid configuration of the room acoustical parameters from a single room impulse response. In section \Romannum{3} we will show how to map a valid configuration to any other valid configuration. Our methods apply to rooms with sets of parallel walls. In addition, we assume the pulses of the direct path, first and second order reflections are available in the observed RIR. Their TOAs can be detected with the methods proposed in \cite{6809961,7096384, ristic2013detection, Risti2014ImprovementOT}. The path lengths of reflections can be computed from the TOAs and the speed of sound. If a set of reflections only reflect on one parallel pair of walls, we say that they belong to the same direction. 

Our proposed method identifies the reflections for each direction and then computes the wall-pair distance with the path lengths of reflections in this direction. The proposed method is based on the image source method and the degeneracy of room impulse responses, which we will discuss in detail in the next section. We first describe our theorem to identify the directions of reflections in the first subsection. In the second subsection, we propose an algorithm to identify reflections. We then compute room geometry and source / receiver positions. In the fourth subsection, we describe how to compute the reflection coefficients.

\subsection{Identifying directions of reflections}

In this subsection, we introduce a theorem to identify the directions of reflections given a set of unlabelled path lengths of the reflections. This theorem is used to classify higher order reflections into two sets, i.e., a multi-direction set and a single-direction set. 

\begin{theorem}
Let $d_{ij}$, $d_i$, $d_j$, and $d_0$ denote the path lengths of the $(O_i+O_j)$-th order reflection that reflects on two directions $i$ and $j$, $O_i$-th order reflections on direction $i$ and $O_j$-th order reflections on direction $j$, and the direct path. Then $d_{ij}^2+d_0^2 = d_i^2+d_j^2$ holds. Vice versa, if there exists a path length of $(O_i+O_j)$-th order reflection $d_{ij}$ that satisfies the equation, the reflections corresponding to path lengths $d_i$ and $d_j$ belong to different directions.
\end{theorem}
\begin{proof}
Let the coordinates of source $s$ and receiver $r$ be $(x_s, y_s, z_s)$ and $(x_r, y_r, z_r)$, respectively, and let the room geometry be $L_x \times L_y \times L_z$. The image source positions can then be represented as $(2m_xL_x + (1-2q)x_s, 2m_yL_y + (1-2j)y_s, 2m_zL_z + (1-2k)z_s)$. We assume the image source of path length $d_i$ to be in the $x$ direction and $d_j$ to be in the $y$ direction. The coordinates of the corresponding image sources are $(2m_xL_x + (1-2q)x_s, y_s, z_s)$, and $(x_s, 2m_yL_y + (1-2j)y_s, z_s)$ respectively. The coordinates of the  $(O_i+O_j)$-th order image source are $(2m_xL_x + (1-2q)x_s, 2m_yL_y + (1-2j)y_s, z_s)$. We can then compute the path length $d_i$ between the image source and the receiver as 
\[
d_i^2 = (2m_xL_x + (1-2q)x_s-x_r)^2 + (y_s-y_r)^2 + (z_s-z_r)^2.
\]
The same formulation holds for the other cases. Formulating $d_i$, $d_j$, $d_{ij}$ and $d_0$ in the same form, it is seen that $d_{ij}^2+d_0^2 = d_i^2+d_j^2$ always holds when $d_i$ and $d_j$ belong to path lengths of reflections in different directions. When they belong to the same direction, the equation is not valid.
\end{proof}

Theorem 1 can also be proved with the parallelogram law as Fig. \ref{fig:para1}. Fig. \ref{fig:para} shows the 2D case for better understanding but it is also valid for the 3D case. From Fig. \ref{fig:para1}, we know the distance between the source and the second order image source equals the distance between two first order image sources, and these two parallelograms share one diagonal. Following the parallelogram law, the sum of the squared side lengths of these two parallelograms are equal, i.e., $d_{ij}^2+d_0^2 = d_i^2+d_j^2$. Fig. \ref{fig:para2} is an example where the higher order path length cannot be written as a function of lower order reflections.

\begin{figure}[h]
     \centering
     \begin{subfigure}[b]{0.5\columnwidth}
         \centering
         \includegraphics[width=\textwidth]{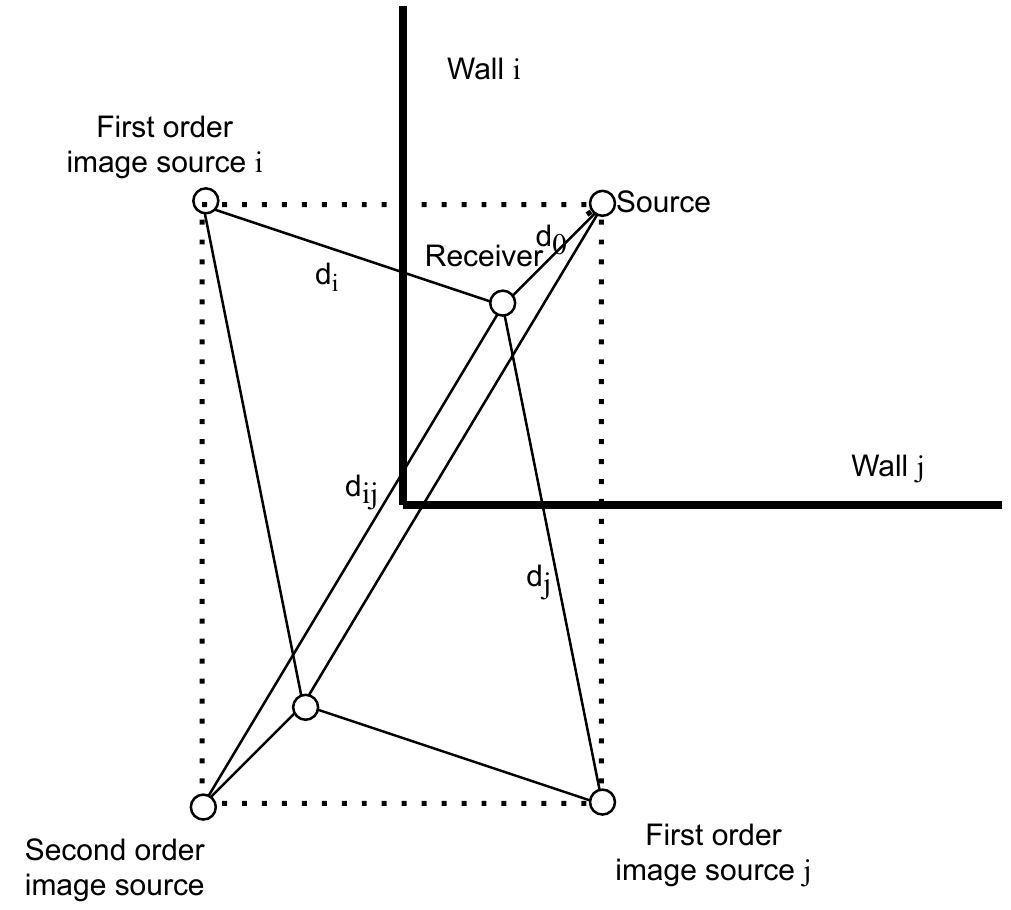}
         \caption{More than one direction}
         \label{fig:para1}
     \end{subfigure}
     \begin{subfigure}[b]{0.35\columnwidth}
         \centering
         \includegraphics[width=\textwidth]{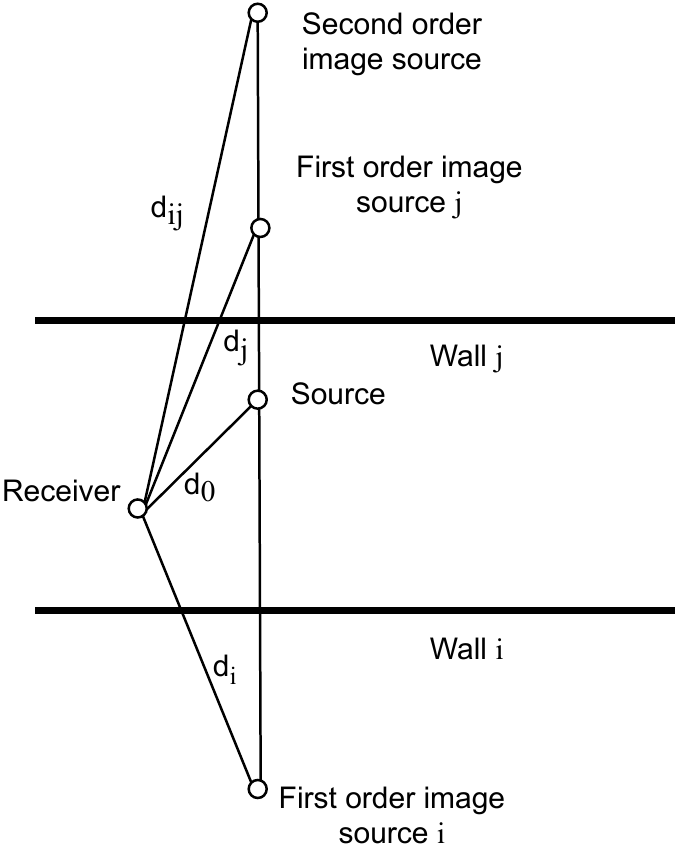}
         \caption{One direction}
         \label{fig:para2}
     \end{subfigure}
     \caption{Second order reflections for the 2D case.}
     \label{fig:para}
\end{figure}
The \textit{multi-direction} set of second (or higher) reflections refers to those that reflect in more than one direction.  The multi-direction set does not provide independent equations since the squared path length can be written as a combination of lower order reflections. The reflections in this set are useless for our purpose and can be pruned out using lower order reflections since they are sure to arrive after the lower order reflections. An example of this is the second order reflection from $(-x_s, -y_s, z_s)$, which reflects on both $x$ and $y$ direction, the squared path length can be written as a combination of the direct path and two first order reflections from $(-x_s, y_s, z_s)$ and $(x_s, -y_s, z_s)$.

The \textit{single-direction} set of second (or higher) order reflections refers to those that reflect along one direction. They can be used to determine the room acoustical parameters. The $N$-th order reflections that reflect along one direction arrive later than the last arrived $(N-1)$-th order reflection in that direction. For example, the second order reflections from $(x_s\pm2L_x, y_s, z_s)$ arrive after the first order reflection from $(2L_x-x_s, y_s, z_s)$. Since the source and receiver are exchangeable, which we will assume below, without loss of generality, we assume $x_r<x_s$, $y_r<y_s$, and $z_r<z_s$ to analyse the sequence of these reflections. As an example, the reflection from $(x_s-2L_x, y_s, z_s)$ arrives earlier than that from $(x_s+2L_x, y_s, z_s)$. 

\subsection{Algorithm of classification of reflections}
\label{sec:algorithm}
In this subsection, we describe our algorithm to classify reflections in detail. We assume we have a set of unlabelled path lengths of reflections in a RIR signal which contain the direct path, the first and second order reflections. We assume the first order reflections are distinguishable. The input set is sorted in ascending order and denoted as $d_k$, $k \in \mathbb{N}$. We aim to classify the path lengths of reflections into five sets, i.e., the first order reflections in the $x$ direction $S_x^1$, the first order reflections in the $y$ direction $S_y^1$, the first order reflections in the $z$ direction $S_z^1$, the second order reflections in a single-direction $S_{\mathrm{single}}^2$, and the multi-direction second order reflections $S_{\mathrm{multi}}^2$. We have $| S_x^1 | =2$, $| S_y^1 | =2$, $| S_z^1 | =2$, $|S_{\mathrm{single}}^2|=6$, and $|S_{\mathrm{multi}}^2|=12$. We introduce a hyperparameter $\delta$ as an error threshold in the path length equation since there might exist a difference between the detected peak position and the theoretical path length. The error threshold depends on the data property and path length.

The first arrived pulse $d_0$ always corresponds to the path length of direct path. Without loss of generality, we assume the path length of the second arrived pulse (first arrived first order reflection) $d_1$ corresponds to $(-x_s, y_s, z_s)$. Let us label $d_1$ as $d_x$. The correctness of this assumption will be explained in Section \Romannum{3}. We iterate over the input set and once a path length of reflection is classified into one set, it will be deleted from the input set. We first find all $d_i$ and $d_j$ that satisfy $d_j \in [ \sqrt{d_i^2+d_x^2-d_0^2} - \delta, \sqrt{d_i^2+d_x^2-d_0^2} + \delta]$ where $1< i < j$. These $d_j$ belong to $S_{\mathrm{multi}}^2$ and without loss of generality we assume the smallest $d_i$ belongs to $S_y^1$. Let us label this path length as $d_y$. We iterate over the remaining found $d_i$ and the remaining path lengths $d_k$ in the input set. For all $d_i$ and $d_k$ that satisfy $d_k \in [ \sqrt{d_i^2 + d_y^2 - d_0^2}-\delta, \sqrt{d_i^2 + d_y^2 - d_0^2}+\delta]$, we have $d_k \in S_{\mathrm{multi}}^2$ and $d_i \in S_z^1$. Then the remaining $d_i$ belongs to $S_y^1$. Till now, we have already found all path lengths of reflections in $S_y^1$ and $S_z^1$. Next, we iterate over the input set again with $d_y$ (this can also be replaced by one of the path lengths in $S_z^1$) to find $d_j$ and $d_i$ that satisfy $d_j \in[ \sqrt{d_i^2+d_y^2-d_0^2}-\delta, \sqrt{d_i^2+d_y^2-d_0^2}+\delta] $. We then have $d_i \in S_x^1$ and $d_j \in S_{\mathrm{multi}}^2$. Lastly, the remaining path length of reflections in the input set are allocated to $S_{\mathrm{single}}^2$.  

The algorithm is robust to mislabelled first or second pulses. If the first two arrived pulses do not correspond to the direct path and the first order reflection in $x$ direction, Theorem 1 does not work for second order reflections and the cardinality of the sets does not match. We can conclude there exist erroneous pulses in the first two arrived pulses. We can then use the path length of the next arrived pulse and select two from these these pulses and repeat the process until the cardinality is correct.

\subsection{Estimation of room geometry and source/receiver positions}

After classifying the reflections, we describe how to compute a valid configuration for the room geometry and positions of source and receiver in this subsection. As discussed, $S_{\mathrm{multi}}^2$ is not useful for this computation. Thus, we only use $S_x^1$, $S_y^1$, $S_z^1$, and $S_{\mathrm{single}}^2$. The arrival sequence of reflections will be explained in section \Romannum{3}.

Our method is based on an iteration of $S_{\mathrm{single}}^2$. This set has six second order reflections in this set, and the coordinates of the image sources are $(x_s\pm2L_x, y_s, z_s)$, $(x_s, y_s\pm2L_y, z_s)$, $(x_s, y_s, z_s\pm2L_z)$. For each of these six elements we perform the following computation. The reflection with the smallest path length in $S_{\mathrm{single}}^2$ has three possible directions. We use each second order reflection candidate, together with the first order reflections in this direction and the direct path, to compute the room geometry and the source and receiver position in this direction. For each second order reflection candidate we then derive the coordinate of another second order reflection in this direction and calculate the corresponding path length, which should be an element of $S_{\mathrm{single}}^2$ for the correct second order reflection candidate. This combination can also be verified with the reflection coefficients in the next subsection.

Let us determine if the hypothesis is correct that a particular distance in $S_{\mathrm{single}}^2$ corresponds to the $x$ direction, i.e., the image source is $(x_s-2L_x, y_s, z_s)$. Together with the path lengths of two first order directions in this direction, with image sources $(-x_s, y_s, z_s)$ and $(2L_x-x_s, y_s, z_s)$, and the path length of direct path from $(x_s, y_s, z_s)$, we can compute $L_x$, $x_s$, and $x_r$. We can then compute the path length of the second order reflection from $(x_s+2L_x, y_s, z_s)$ in this set and if this is consistent with the initial hypothesis, then we have verified it. If so, we also computed the second second order reflection in this direction. This procedure allows us to find the second-order pulses in each direction.

The proposed method is relatively robust in two aspects. Firstly, generalising this algorithm to include additional higher order reflections can improve the robustness of the algorithm. Theorem 1 can also be applied on higher order reflections to classify directions. We can apply the proposed method to the higher order reflections that reflect along one direction. As a result, the higher order reflections can be used to verify the solution of room acoustical parameters to improve robustness. Secondly, the proposed method is robust to additional peaks that do not belong to any reflection. With the algorithm in the previous subsection, these pulses will be misclassified into the set that contain the higher order reflections that reflect along the same direction. Since another second order reflection that also reflects along this direction does not exist, we know this peak is erroneous.

Independent calculation on each direction is an advantage of our proposed method. Since the calculation for each direction is separable, the method can also be applied to some special cases. We can estimate the distance between parallel walls and the source/receiver position along this pair of parallel walls whether the walls on other directions are very distant (such as in a hallway) or affected by furniture. The independent calculation for each direction makes our method work for non-shoebox shaped rooms with sets of parallel wall pairs as long as the required reflections are available. An example is a room with a sloped ceiling. Such a room has two pairs of parallel wall and a pair of non-parallel walls. For the second order reflections between the vertical wall and the floor, since they form a right angle, the second order reflections will be pruned out by the first order reflections. For the second order reflections between the ceiling and the remaining walls, since they do not follow Theorem 1, they will not be pruned out and will be classified into $S_2^{\mathrm{single}}$. However, they will be recognised as erroneous pulses since there is no valid solution of room acoustical parameters.

\subsection{Estimation of reflection coefficients}

In addition to the estimation of the room geometry and positions of source and receiver, we can estimate reflection coefficients of each wall with first order reflections. With the image source method, from \eqref{RIRIS}, we know that the amplitude of each reflection is only related to the distance between the image source and receiver and the reflection coefficients. For the first order image sources, we know their path lengths and the true amplitude from the RIR signal. Since they only reflect once, the amplitude is equal to $\frac{\beta_i}{4\pi d_i}$, where $\beta_i$ is the corresponding reflection coefficient. Hence we can compute the reflection coefficient for each wall. With the reflection coefficients on each wall and the computed parameters, we can then compute the amplitude of second order reflections in each direction. If the computed amplitudes match the measured amplitude, then this confirms the value of the reflection coefficients.

\section{Room impulse response degeneracy analysis}

In Section \ref{sec:paraestimate}, we discussed how to compute room acoustical parameters from a room impulse response. However, the room configuration, including room geometry, positions of source and receiver, reflection coefficients, and the coordinate system, is not unique for a particular RIR. Fig. \ref{2Dconf} is a 2D example where eight different configurations result in an identical RIR. In addition, if we exchange the positions of source and receiver, the RIR will also not change. Hence, for a 2D room, $16$ configurations can result in an identical RIR. Except for this degeneracy, the TOA of reflections of a RIR is unique with respect to a room configuration. 
\begin{figure}[h]
\centerline{\includegraphics[width=0.9\columnwidth]{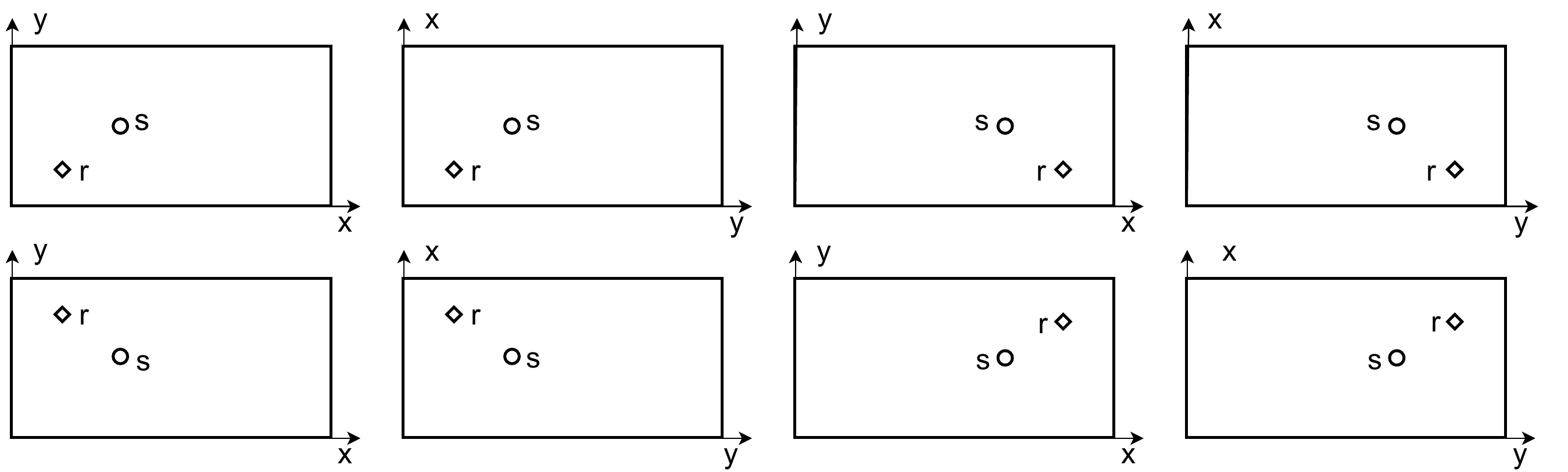}}
\caption{Different 2D configurations resulting in an identical RIR where $s$ and $r$ denote source and receiver.}\label{2Dconf}
\end{figure}

Next, we analyse the degeneracy for a 3D room and determine the coordinate system based on the first order reflections. Let us assume the first reflection is the pulse that reflects on the wall $x = 0$. Since the first pulse can reflect on any wall, this gives us six-fold degeneracy. Similarly, the next arriving first order reflection that is not in the $x$ direction introduces a four-fold degeneracy by assuming it reflects on the wall $y=0$. Finally, the third arriving first-order reflection that is not in the $x$ and $y$ directions has two-fold degeneracy by assuming it reflects on the wall $z=0$. This results in an overall 48-fold degeneracy. In addition, the coordinates of source and receiver are exchangeable, which results in a total of $96$-fold degeneracy. 

\begin{figure}[h]
\centerline{\includegraphics[width=\columnwidth]{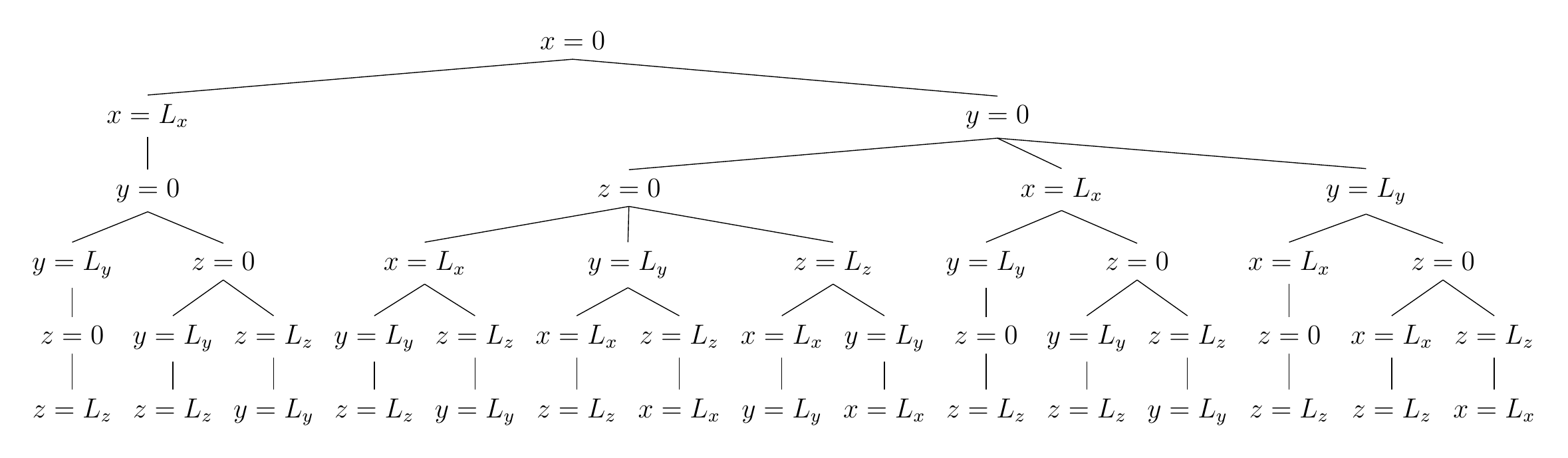}}
\caption{The sequence of first order reflections for the 3D case.}\label{3D_tree}
\end{figure}
We discuss one case as an example how the degeneracies appear in a practical setup since other cases are similar. The case corresponds to one particular branching pathway in Fig. \ref{3D_tree}, which shows the possible orderings of first-order reflections. As discussed, the first arrived first order reflection reflects on $x = 0$, which introduces six-fold degeneracy. We note that the second arrived first order reflection can reflect on $x = L_x$ or in the $y$ or $z$ direction. We consider the case where the second arrived first order reflection reflects on $x = L_x$. We then define the third arrived first order reflection to reflect on $y=0$, which introduces four-fold degeneracy. Then the fourth arrived first order reflection will reflect on $y = L_y$ or in the $z$ direction. We consider that the fourth arrived first order reflection reflects on $y = L_y$. We then define the fifth arrived first order reflection to reflect on $z = 0$ and the sixth arrived first order reflection to reflect on $z = L_z$. Assuming the fifth reflection to reflect on $z=0$ introduces two-fold degeneracy. The exchange of the coordinates of source and receiver results in an overall $96$-fold degeneracy for this branching pathway. For other pathways we find the same result.

Mapping from one of the degenerate solutions to another is straightforward. The methods consist of coordinates exchange of source and receiver, exchange of $x$, $y$, and $z$ coordinates, and symmetry with respect to $x = \frac{L_x}{2}$, $y = \frac{L_y}{2}$, or $z = \frac{L_z}{2}$. For an audio-only environment, it does not matter which case to choose. However, for an audio-visual environment,  it is necessary to match one of the degeneracies with the visual scene for an acceptable experience. We need to use the available visual cues to determine which case and what combination of methods to use, which is out of the scope of this paper.

\section{Experiments}

We evaluated our algorithm in section \ref{sec:algorithm} with the RIRs generated by the image source method without and with scattering \cite{raven}. The scattering coefficient \cite{scattering} was set to be $0.1$ and $0.2$ to compare with the RIR with specular reflections only. The rooms were assumed to be rectangular and empty. We randomly generated $1000$ rooms. Each dimension of room geometry was uniformly distributed between $2 \times 5 \times 7$ m and $4 \times 10 \times 11$ m. One source and one receiver were randomly placed inside the room. The speed of sound was set to $c=343$ m/s. The sampling frequency was set to $fs = 44100$ Hz. The length of each RIR was $127890$ samples, which corresponds to a $2.9$ s signal. The reflection coefficients of the walls were simulated as iid between $0$ and $1$. We assume the bandwidth of the RIR equals to the Nyquist bandwidth and the error threshold was set to $\delta = 5/fs$ .

We recorded the RMSE of the estimated room geometry, source/receiver position and reflection coefficients. The results are shown in Table \ref{tab1}. We also recorded the failed cases when it fails to calculate the room acoustical parameters with the detected reflections and when the RMSE on one of the room acoustical parameters is larger than $1$. The failed proportion of scattering coefficient $0$, $0.1$, and $0.2$ are $0.7\%$, $6.8\%$, and $18.9\%$ respectively. The failed portion can be reduced by repeating the experiments with re-measured RIRs. The experimental results show that larger scattering coefficients will decrease the estimation accuracy of the room acoustical parameter estimation.
\begin{table}[h]
\centering
\caption{RMSE of the estimation of room acoustical parameters.}
\label{tab1}
\begin{tabular}{| c | c |  c | c |}
\hline
Scattering coefficient &  $0$ (No scattering) & $0.1$  & $0.2$\\
\hline
Room geometry (m) & $0.0505$ & $0.2286$ & $0.2633$ \\
\hline
Receiver position (m) & $0.1371$ & $0.1803$ & $0.3088$  \\
\hline
Source position (m) & $0.1739$ & $0.2329$ & $0.5792$  \\
\hline
Reflection coefficients & $0.1386$ & $0.1401$ & $0.1478$  \\
\hline
\end{tabular}
\end{table}

\section{Conclusion}
In this paper, we proposed an analytical method to analyse RIRs and derive the room geometry, the positions of source and receiver, and reflection coefficients simultaneously. We only require the times of arrivals of the direct path, first order reflections and second order reflections of a single RIR between a source and a receiver inside a room with parallel walls and do not require any other information. The proposed method is robust to erroneous pulses and non-specular reflections. However, the proposed method has limitations for real-world applications since it depends on the TOA estimation of early reflections. Further work can focus on improving robustness to real-world applications.

\bibliographystyle{IEEEtran}
\bibliography{refs}

\end{document}